\documentclass[a4paper,12pt]{article}
\usepackage{amsmath,amssymb,amsfonts,amsthm}
\newtheorem{theorem}{Theorem}[section]
\newtheorem{lemma}[theorem]{Lemma}

\newcommand{\Rt}{\mathbb{R}^3}

\title{Extreme Bowen-York initial data}

\author{Sergio Dain$^{1,2}$ and Mar\'\i a E. Gabach Cl\'ement$^1$\\
\\
$^1$Facultad de Matem\'atica, Astronom\'{i}a y F\'{i}sica, \\
Universidad Nacional de C\'ordoba,\\
 Ciudad Universitaria, (5000) C\'ordoba, Argentina.  \\
\\
$^{2}$Max Planck Institute for Gravitational Physics,\\
  (Albert Einstein Institute), Am M\"uhlenberg 1,\\
  D-14476 Potsdam Germany. }

\begin{document}

\maketitle

\abstract{The Bowen-York family of spinning black hole initial data
  depends essentially on one, positive, free parameter.  The extreme
  limit corresponds to making this parameter equal to zero. This choice
  represents a singular limit for the constraint equations. We prove
  that in this limit a new solution of the constraint equations is
  obtained. These initial data have similar properties to the extreme
  Kerr and Reissner-N\"ordstrom black hole initial data. In
  particular, in this limit one of the asymptotic ends changes from
  asymptotically flat to cylindrical.  The existence proof is
  constructive, we actually show that a sequence of Bowen-York data
  converges to the extreme solution.}

\section{Introduction}
\label{sec:introduction}


The Kerr-Newman black hole depends on three parameters, $m$, $q$ and
$J$, the mass, the electric charge and the angular momentum of the
spacetime respectively.  They satisfy the following well known
inequality
\begin{equation}
  \label{eq:39}
  m^2 \geq q^2 + \frac{J^2}{m^2}.
\end{equation}
This inequality can be written in the following
form in which the mass appears only in the left hand side of the equation
and on the right hand side we have all the `charges'
\begin{equation}
  \label{eq:37}
  m^2 \geq \frac {q^2+ \sqrt{q^4+ 4J ^2}}{2}.
\end{equation}
The extreme Kerr-Newman black hole is defined by the equality in
(\ref{eq:37})
\begin{equation}
  \label{eq:41}
  m^2 = \frac {q^2+ \sqrt{q^4+ 4J ^2}}{2}.
\end{equation}
For fixed values of $q$ and $J$, we can interpret the extreme black
hole as the black hole with the minimum mass. In other words, the
extreme black hole has the maximum amount of charge and angular
momentum per mass unit allowed for given values of $q$ and $J$. This
variational interpretation of extreme black holes generalizes to
non-stationary, axially symmetric, black holes (\cite{Dain06c},
\cite{Dain:2007pk}, \cite{Dain05e}). It is convenient to define a
parameter $\mu$ which measures how far a black hole is with respect to
the extreme case. In the stationary case, assuming that $m$, $q$ and
$J$ satisfy inequality (\ref{eq:37}), $\mu$ is given by
\begin{equation}
  \label{eq:40}
  \mu=\sqrt{ m^2 - \frac {q^2+ \sqrt{q^4+ 4J ^2}}{2}}.
\end{equation}
Note that $\mu$ has unit of mass.  The extreme limit corresponds to
$\mu = 0$. For Schwarzschild solution we have $\mu=m$.

In the extreme limit the global structure of the spacetime changes
(see \cite{Carter73}).  Particularly relevant for the study of black
holes as an initial value problem is the change in the structure of
Cauchy surfaces, and hence initial data set, in this limit.  The
slices $t=constant$ in Boyer-Lindquist coordinates represent Cauchy
surfaces for the Kerr-Newman black hole. For $\mu>0$ these slices have
two isometrical asymptotically flat ends. In the extreme limit $\mu=0$
one of the ends changes from asymptotically flat to cylindrical. Also,
for $\mu>0$ the Cauchy surfaces contain an apparent horizon (in this
case, due to the symmetry, it is also a minimal surface). In the
extreme case they do not contain any apparent horizons or minimal
surfaces.

We can characterize a black hole spacetime by an initial data
set. Then, it is possible to define an analog to the extreme limit
discussed above for more general (in particular, non-stationary) black
holes families. In \cite{Dain:2008ck} the extreme limit for the
Bowen-York family of spinning black holes initial data set was
defined.  The Bowen-York initial data set \cite{Bowen80} describes
non-stationary, axisymmetric, black holes with angular
momentum. Having fixed the angular momentum, the Bowen-York family
depends on one parameter, which is the analog of the $\mu$ parameter
defined in \eqref{eq:40}. As for the Kerr-Newman black hole, the
extreme limit in this case also corresponds to $\mu \to 0$. The
problem is that these data are not given explicitly. They are
prescribed as solutions of a non-linear elliptic equation
(essentially, the Hamiltonian constraint) with appropriate boundary
conditions. For the case $\mu >0$ it is well known that this equation
has a unique solution. However, the value $\mu=0$ represents a
singular limit for this equation. In \cite{Dain:2008ck} this limit was
explored numerically. The numerical calculations indicate that in the
limit a new solution is obtained (see also \cite{Lovelace:2008tw}).
The purpose of this article is to prove this. Namely, we will prove
that the sequence $\mu \to 0$ of Bowen-York spinning black hole data
converges to a limit solution. We call this new solution of the
constraint equations the extreme Bowen-York data.  We also prove that
the solution (as it was showed numerically in \cite{Dain:2008ck}) has
a similar behavior to the extreme Kerr-Newman initial data discussed
above.

The Bowen-York spinning black hole initial data has been extensively
used in numerical relativity (see the review article
\cite{Cook00}). The extreme Bowen-York data constructed here represent
the data with the maximum amount of angular momentum per mass unit in
this family and hence they have applications in astrophysical
scenarios in which highly spinning black holes are relevant (see the
discussion in \cite{Dain:2008ck}, \cite{Lovelace:2008tw} and
references therein).

As a final comment, we mention that asymptotically flat Riemannian
manifolds have been extensively studied in General Relativity in
connection with the constraint equations (see the review article
\cite{Bartnik04b}). On the other hand very little is known about
manifolds with cylindrical ends which appear naturally in
extreme black holes. The solution presented here represents a
non-stationary and non-trivial example of such manifolds.

The article is organized as follows. In section \ref{sec:main-result}
we present our main result given by theorem \ref{t:1} and we discuss
its implications. The proof of this theorem is split in section
\ref{sec:monot-sequ}, \ref{sec:bounds-sequence} and
\ref{sec:convergence-sequence}. Possible generalizations and further
studies are discussed in section \ref{sec:final-comments}. Finally, in
the Appendix we give the explicit expression of a lower bound for the
solution that can be useful in numerical calculations.

\section{Main result}
\label{sec:main-result}

Let us review the Bowen-York spinning black hole initial
data \cite{Bowen80} with `puncture' boundary conditions
\cite{Brandt97b}. The 3-dimensional manifold is given by $\Rt\setminus
\{0\}$. On  $\Rt\setminus
\{0\}$ the metric $h_{ij}$ and second fundamental form $K_{ij}$ are given by
\begin{equation}
  \label{eq:3}
  h_{ij}=\Phi^4\delta_{ij}, \quad K_{ij}=\Phi^{-2}\sigma_{ij},
\end{equation}
where $\delta_{ij}$ is the flat metric and the tensor $\sigma_{ij}$ is
given by
\begin{equation}
  \label{eq:sigma}
  \sigma_{ij}=\frac{6}{r^3}n_{(i}  \epsilon_{j)kl} J^k n^l,
\end{equation}
where $r$ is the spherical radius, $n^i$ the corresponding radial
unit normal vector, $\epsilon_{ijk}$ the flat volume
element  and $J_k$ an arbitrary  constant vector. In this equation the
indices are moved with the flat metric $\delta_{ij}$.

The conformal factor $\Phi$ satisfies the following non-linear
elliptic equation in $\Rt\setminus \{0\}$
\begin{equation}
\label{eq:BYPhi}
\Delta{\Phi}=  F(x,\Phi),
\end{equation}
where
\begin{equation}
  \label{eq:F}
  F(x,\Phi)=-\frac{9J^2\sin^2\theta}{4r^6\Phi^7},
\end{equation}
and $J^2=J_iJ_j\delta^{ij}$, $\Delta$ is the flat Laplacian
and $x$ denotes spherical coordinates $(r,\theta)$.

Boundary conditions for black holes are prescribed as follows. For a
given parameter $\mu >0$ define
the function  $u_\mu$ on $\Rt$, by
\begin{equation}
  \label{eq:Phimu}
  \Phi_\mu:=1+\frac{\mu}{2r}+u_\mu.
\end{equation}
Inserting this definition in equation (\ref{eq:BYPhi}) we obtain the
following equation for  $u_\mu$
\begin{equation}
\label{ecby}
\Delta u_{\mu}=F(x,\Phi_\mu),
\end{equation}
where
\begin{equation}
  \label{eq:5}
  F(x,\Phi_\mu)= - \frac{9J^2\sin^2\theta}{4r^6(1+\frac{\mu}{2r}+u_{\mu})^{7}}.
\end{equation}
Then, equation (\ref{ecby}) is solved in $\Rt$ subject to the
asymptotic behavior
\begin{equation}
  \label{condecby}
  u_{\mu}\rightarrow 0 \quad \text{ as } r\rightarrow\infty.
\end{equation}
For every $\mu>0$ there exists a unique solution of (\ref{ecby}) such
that it satisfies (\ref{condecby}). A proof of this was given in
\cite{Brandt97b} based on \cite{Cantor79}. It is also possible to
prove this result using a suitable adapted version of the sub and
supersolution theorem presented in \cite{Choquet99} or using a
compactification of $\Rt$ like the existence theorems in \cite{Beig94}
\cite{Dain99}.

Note that equation (\ref{ecby}) depends, in principle, on two
parameters, $J$ and $\mu$. There exists however a scale invariance
for this equation (see \cite{Dain:2008ck}), and hence the solution
depends non trivially only on one parameter. We chose to fix $J$ and vary $\mu$.

In the rest of the article we will denote by $u_\mu$
the unique solution of (\ref{ecby}), with boundary condition
(\ref{condecby}) for any given $\mu>0$. We have that $u_\mu \geq 0$ and $u_\mu\in
C^{2,\alpha}(\Rt)$, where $C^{k,\alpha}(\Rt)$ denotes H\"older
spaces (see, for example, \cite{Gilbarg} for definition and properties
of these functional spaces).

The total angular momentum of the data is given by $J$ and the total
mass $m$  is given by
\begin{equation}
  \label{eq:42}
m= \mu + \frac{1}{4\pi} \int_{\Rt}
\frac{9J^2\sin^2\theta}{4r^6(1+\frac{\mu}{2 r}+u_{\mu})^{7}} \, dx.
\end{equation}
Note that the mass can not be a priori explicitly calculated as a function of
$\mu$ and $J$ since it involves the solution $u_\mu$.

As we said in the introduction, we are interested in studying the limit
\begin{equation}
  \label{eq:32}
\lim_{\mu \to 0} u_\mu .
\end{equation}
The corresponding equation becomes
\begin{equation}
\label{ecbye}
\Delta u_0=-\frac{9J^2\sin^2\theta}{4r^6(1+u_0)^{7}}.
\end{equation}
We remark that when $\mu>0$, the right hand side of (\ref{ecby}) is bounded
in $\Rt$ (this is of course related with the fact the the solution
$u_\mu$ is regular at the origin for $\mu >0$). Whereas in the extreme
case, $\mu=0$, it becomes singular at the origin, and hence we can not
expect the solution $u_0$ to be regular at the origin.

The following theorem constitutes the main result of the present
article. To formulate the theorem we will use weighted Sobolev spaces,
 denoted by $H'^{2,\delta}$, defined in \cite{Bartnik86} (see equation
(\ref{eq:38}) in section \ref{sec:convergence-sequence}).
\begin{theorem}
\label{t:1}

(i) There exists a solution $u_0$ of equation (\ref{ecbye}) in
$\Rt\setminus \{0\}$ such that  $u_0\in C^\infty (\Rt\setminus
\{0\})$ and $u_0$ satisfies the following bounds
\begin{equation}
  \label{eq:1}
 u^-_0 \leq u_0 \leq u^+_0,
\end{equation}
where the functions $u^+_0$ and $u^-_0$ are explicitly given by
\begin{equation}
  \label{eq:supsolm}
  u^+_0=\sqrt{1+\frac{|q|}{r}}-1, \quad |q|=\sqrt{3|J|},
\end{equation}
and
\begin{equation}
\label{eq:subsolm}
u^-_0 =Y_{00}(\theta)I_1(r) -\frac{Y_{20}(\theta)}{5^{3/2}} I_2(r).
\end{equation}
Here $Y_{00}$ and $Y_{20}$ are spherical harmonics (see equation
(\ref{eq:sphar}))
and the radial functions $I_1(r)$ and $I_2(r)$ are given explicitly in
the Appendix (equations (\ref{eq:I1mu1}) and (\ref{eq:I2mu2})).
\par
(ii) In addition, we have that  $u_0\in H'^{2,\delta}$ for
$-1<\delta<-1/2$ and $u_0$ is the
limit of the sequence
\begin{equation}
  \label{eq:4}
  \lim_{\mu \to 0}  u_\mu=u_0,
\end{equation}
in the norm $H'^{2,\delta}$.
\end{theorem}

The bounds (\ref{eq:1}) obtained in part (i) of theorem \ref{t:1} imply
\begin{equation}
  \label{eq:2}
  u_0 = O(r^{-1}) \text{ as } r \to \infty, \quad u_0=O(r^{-1/2})\text{
    as } r \to 0.
\end{equation}
These bounds show that the limit solution $u_0$ behaves different near
the origin from the sequence's members $u_\mu$.  This behavior confirms the
numerical calculations presented in \cite{Dain:2008ck} and
\cite{Lovelace:2008tw}.  This is also related to the change of one of
the ends from asymptotically flat to cylindrical in the extreme
limit. To see this, we calculate the area of the 2-surfaces
$r=constant$ with respect to the physical metric $h_{ij}$ defined in
(\ref{eq:3}). The area $A$ is given by
\begin{equation}
  \label{eq:22}
  A_\mu(r) = 2\pi r^2 \int_0^\pi \Phi_\mu^4 \sin \theta d\theta.
\end{equation}
It is well known that for $\mu>0$ the surface $r=\mu/2$ is a minimal
surface. Also, for $\mu>0$ we have
\begin{equation}
  \label{eq:23}
\lim_{r\to \infty}  A_\mu(r)=\lim_{r\to 0}  A_\mu(r)=\infty,
\end{equation}
which reflects the fact that the data has two asymptotically flat ends. Moreover,
these asymptotic regions are isometrical and are connected by the
minimal surface at $r=\mu/2$.
The situation changes in the extreme limit. Using the bounds
(\ref{eq:1}) we can obtain the following bounds for the area in this limit
\begin{equation}
  \label{eq:24}
  0< 2.37\pi |J| \leq  A_0(0)\leq   12\pi |J|.
\end{equation}
We see that the point $r=0$ has finite, non-zero area. This shows that
$r=0$ is not an asymptotically flat end. It is a cylindrical end
similar to the one present in extreme Kerr and extreme Reissner N\"ordstrom.
On the other hand, the behavior as $r\to \infty$ is identical in both the
non-extreme and the extreme cases. That is, this end is always
asymptotically flat.

Note that in part (i) of theorem \ref{t:1} nothing is said about the
behavior of the derivatives of $u_0$ near the origin and the fall off
near infinity. The behavior of the derivatives of $u_0$ in these
regions is analyzed in part (ii) with the weighted Sobolev spaces. In
particular, these spaces provide a norm for the convergence of the
sequence and its derivatives in $\Rt$.

Finally, let us mention three important points which we were unable to
analyze at the moment. The first one is uniqueness of the solution
$u_0$.  We have not proven that this solution is unique in
$H'^{2,\delta}$ or in other suitable functional space.  The second
point is related with the behavior of the total mass in the sequence
$u_\mu$. The numerical calculations show that the mass decrease as
$\mu \to 0$ (see\cite{Dain:2008ck} \cite{Dain02c}). This is of course
the main reason why we call this solution the extreme Bowen-York
data. However, we did not prove this analytically.  Finally the third
point is concerned with the existence of minimal surfaces and
horizons. We believe that the extreme solution does not have any
minimal surface or apparent horizon (in analogy with the extreme
Kerr-Newman black hole). This is also indicated in numerical
calculations. But we were unable to show this.

The proof of theorem \ref{t:1} falls naturally into three parts
presented in section \ref{sec:monot-sequ}, \ref{sec:bounds-sequence}
and \ref{sec:convergence-sequence}. The plan of the proof is presented
below.
\begin{proof}

  We first prove that the sequence $u_\mu$ is pointwise monotonically
  increasing as $\mu$ decrease. This is proved in lemma \ref{l:monu}.
  Then, we show that there exists a function $u_0^+$, independent of
  $\mu$, which is an upper bound to this sequence for all $\mu$. See
  theorem \ref{l:rn}. This theorem constitutes the most important part
  of the proof.  From this upper bound we construct a lower bound
  $u_0^-$ in lemma \ref{l:rnsub}.  Combining these lemmas and using
  standard elliptic estimates for the Laplacian on open balls which do
  not contain the origin we prove that the limit (\ref{eq:4}) exists
  and $u_0$ is smooth outside the origin. See lemma \ref{l:icon}.
  This proves the part (i) of the theorem.  Finally, part (ii) is
  proved in lemma \ref{l:ws}.
\end{proof}
\section{Monotonicity}
\label{sec:monot-sequ}
The function $F(x,\Phi)$ defined by (\ref{eq:F}) is
non-decreasing in $\Phi$.  This fact, together with the maximum
principle for the Laplace operator, will allow us to prove the
monotonicity of the sequence $u_\mu$ with respect to the parameter $\mu$.

The non-decreasing property of $F$ is conveniently written in the
following way. Let $\Phi_1$, $\Phi_2$ be positive functions such that
$\Phi_1\geq\Phi_2$, then we have
\begin{equation}
  \label{eq:Fdif}
F(x,\Phi_1)-F(x,\Phi_2) = (\Phi_1-\Phi_2) H(\Phi_2,\Phi_1)\geq0,
\end{equation}
 where we have defined the function $H(\Phi_2,\Phi_1)=H(\Phi_1,\Phi_2)$ as
\begin{equation}
\label{eq:H}
H(\Phi_2,\Phi_1)=\frac{9J^2\sin^2\theta}{4r^6}\sum_{i=0}^{6}
\Phi_{1}^{i-7}\Phi_{2}^{-1-i}\geq0,
\end{equation}
and we have used the following elementary identity  for real numbers $a$ and $b$
\begin{equation}
  \frac{1}{a^p}-\frac{1}{b^p}=(b-a)
  \sum_{i=0}^{p-1}a^{i-p}b^{-1-i}.
\end{equation}
In our case the functions $\Phi$ are given by (\ref{eq:Phimu}) with
$\mu\geq 0$, and since $u_\mu\geq 0$ for $\mu >0$, from (\ref{eq:Phimu}) we
obtain an  upper bound for $H$
\begin{equation}
  \label{eq:Hupmu}
  |H(\Phi_{\mu_2} , \Phi_{\mu_1})|\leq \frac{9J^2
    r^2\sin^2\theta}{4}\sum_{i=0}^{6}\left(r+\frac{\mu_1}{2}\right)^{i-7}
  \left(r+\frac{\mu_2}{2}\right)^{-1-i},
\end{equation}
which shows that $H$ is bounded in $\Rt$ if $\mu_1, \mu_2
>0$. Taking $\mu_1=\mu_2=0$, from (\ref{eq:Hupmu}) we obtain the
following bound which is independent of $\mu$
\begin{equation}
  \label{eq:13}
  |H(\Phi_{\mu_2} , \Phi_{\mu_1})|\leq \frac{63J^2\sin^2\theta}{4r^6}.
\end{equation}
Note that this bound diverges at the origin.

The main result of this section is summarized in the following lemma.
\begin{lemma}
\label{l:monu}
 Assume $\mu_1\geq \mu_2>0$ then we have $u_{\mu_1}(x)\leq u_{\mu_2}(x)$ in
 $\Rt$.
\end{lemma}
\begin{proof}

Define  $w$ by
\begin{equation}\label{w}
w(x)=u_{\mu_2}(x)-u_{\mu_1}(x).
\end{equation}
Using equation (\ref{ecby}), we obtain that $w$ satisfies the equation
\begin{equation}
\label{ecw}
\Delta w=F(x,\Phi_{\mu_2})-F(x,\Phi_{\mu_1}).
\end{equation}
We use (\ref{eq:Fdif}) to write this equation in the following form
\begin{equation}\label{ww}
\Delta w-w H (\Phi_{\mu_2} , \Phi_{\mu_1})=\frac{\mu_2-\mu_1}{2r}H(\Phi_{\mu_2} ,
\Phi_{\mu_1}),
\end{equation}
where $H$ is given by (\ref{eq:H}). Since $H\geq 0$ and by hypothesis
we have $\mu_2-\mu_1\leq 0$, then the right hand side of (\ref{ww}) is
 negative. We also have that $w\to 0$ as $r\to \infty$ (because of
(\ref{condecby})). Hence, we can apply the Maximum Principle for the
Laplace operator to conclude that $w\geq0$ in $\Rt$.  We use a version
of the Maximum Principle for non-bounded domains given in
\cite{Choquet99}. We emphasize that this classical version of the
maximum principle can be applied in the present case because $w$ is 
$C^{2,\alpha}$ and $H$ is bounded in $\Rt$ when $\mu_1,\mu_2 > 0$.
\end{proof}

Remarkably, the sequence $\Phi_\mu$ has the opposite behavior as the
sequence $u_\mu$, namely  $\Phi_\mu$ is increasing with respect to
$\mu$. This is proved in the following lemma.
\begin{lemma}
\label{l:monphi}
 Assume $\mu_1 > \mu_2>0$ then we have $\Phi_{\mu_1}(x)\geq \Phi_{\mu_2}(x)$ in
 $\Rt\setminus \{0\}$.
\end{lemma}
\begin{proof}
The proof is similar to the one in the  previous lemma, however since $\Phi_\mu$ is
singular at the origin we need to exclude this point from the domains.
In order to handle this, let $\Omega$ be defined as
$\Rt\setminus B_\epsilon$ where $B_\epsilon$ is a small ball of
radius $\epsilon$ centered at the origin.
As before, we define $w$ as the difference
\begin{equation}
\label{eq:wPP}
w=\Phi_{\mu_2}-\Phi_{\mu_1}=\frac{\mu_2-\mu_1}{2r} + u_{\mu_2}-u_{\mu_1} .
\end{equation}
Then, we have
\begin{equation}
\label{eq:wPhi}
\Delta w -wH(\Phi_{\mu_1},\Phi_{\mu_2}) =0,
\end{equation}
where $H$ is given by (\ref{eq:H}).  Since $u_\mu$ is bounded in $\Rt$
for $\mu>0$ then the first term in the right hand side of
(\ref{eq:wPP}) will dominate for sufficiently small $r$. Hence, for
$\mu_1 > \mu_2>0$ there exists $\epsilon$ sufficiently small such that
$w$ is negative on $\partial B_\epsilon$.  Consider equation
(\ref{eq:wPP}) on $\Omega$. The function $w$ is negative on $\partial
B_\epsilon$ and it goes to zero at infinity. Hence, we can apply the
Maximum Principle in $\Omega$ to obtain $w\leq 0$. In this case we
need to slightly modify the version of the maximum principle given in
\cite{Choquet99} to include the inner boundary $\partial
B_\epsilon$. This modification is however straightforward.
\end{proof}

\section{Bounds}
\label{sec:bounds-sequence}
In this section we give bounds for the sequence $u_\mu$ . The main
result of the section is given by theorem \ref{l:rn} where we
construct an upper bound $u_0^+$, based on the Reissner N\"ordstrom
black hole initial data, which does not depends on $\mu$.  The lower
bound is then directly constructed using this upper bound in lemma
\ref{l:rnsub}.

The Reissner N\"ordstrom black hole will play an important role in what
follows. Let us  review it. The Reissner N\"ordstrom metric is
characterized by two parameters: the mass $m$ and the electric charge
$q$. This metric describes a black hole if $|q|\leq m$. When
$|q|= m$ the solution is called the extreme  Reissner N\"ordstrom
black hole.  Take a
 slice $t=constant$ in the canonical coordinates and let $r$ be the
 isotropical radius on this slice. The intrinsic metric on the slice is
 conformally flat, i.e.  it has the  form (\ref{eq:3}) where the conformal
 factor is denoted by $\Phi_\mu^+$ (the reason for the $+$ in the
 notation will became clear later on) and it is explicitly  given by
 \begin{equation}
   \label{eq:26}
\Phi_\mu^+=   \sqrt{1+\frac{m}{r}+\frac{\mu^2}{4r^2}},
 \end{equation}
where the parameter $\mu$ is defined in terms of $m$ and $q$ by
(\ref{eq:40}) (with $J=0$), that is
\begin{equation}
  \label{eq:9}
  m=\sqrt{\mu^2+q^2}.
\end{equation}
Note that, when $q$ is fixed, then $m$ decreases as $\mu$ goes to
zero. We also define the function $u^+_{\mu}(x)$ by
 \begin{equation}
   \label{eq:25}
  \Phi_\mu^+=1+\frac{\mu}{2r} +u^+_{\mu}(x),
 \end{equation}
that is, we have
\begin{equation}
\label{solrn}
u^+_{\mu}(x)=\sqrt{1+\frac{m}{r}+\frac{\mu^2}{4r^2}}-1-\frac{\mu}{2r}.
\end{equation}
The extreme limit corresponds to $\mu=0$, in this limit the solution is
denoted by $u^+_{0}$, we have
\begin{equation}
  \label{eq:extrn}
  u^+_{0}(x)=\sqrt{1+\frac{|q|}{r}}-1.
\end{equation}
As a consequence of the constraint equations the function $u^+_{\mu}$
satisfies
\begin{equation}
\label{ecrn}
\Delta u^+_{\mu}=-\frac{q^2}{4r^4 \left(  \Phi_\mu^+ \right)^3}.
\end{equation}

We have $u^+_{\mu} \geq 0$. From the explicity expresion (\ref{solrn})
we deduce that the sequence $u^+_{\mu}$ is increasing as
$\mu \to 0$ and it is bounded by the extreme solution $u^+_0$, that is
\begin{equation}
u^+_{\mu}(x) < u^+_{0}(x),
\end{equation}
for all $\mu>0$. Also, $u^+_{\mu}(x)$ is smooth on $\Rt\setminus\{0\}$
and, for $\mu>0$, we have $u^+_{\mu}\in C^1(\Rt)$ (but it is not $C^2$
at the origin). The values of the function and its derivative at the
origin are given by
\begin{equation}
u^+_{\mu}(0)=\frac{m}{\mu}-1 \quad \frac{du^+_{\mu}}{dr}(0)=-\frac{q^2}{\mu^3}.
\end{equation}
Note that both values diverge as $\mu \to 0$. In fact the limit
function $u^+_0$ diverges as $r^{-1/2}$
near the origin. We want to prove that a similar behavior occurs for the
Bowen York case.

The following constitutes the main result of this section.
\begin{theorem}
\label{l:rn}
Assume that
\begin{equation}
\label{condicion}
|q|\geq \sqrt{3|J|}.
\end{equation}
Then for all $\mu>0$  we have
\begin{equation}
u_{\mu}(x)\leq u^+_{\mu}(x) <  u^+_{0}(x),
\end{equation}
where  $u^+_\mu$  and $u^+_0$ are given by  (\ref{solrn}) and
(\ref{eq:extrn}) respectively.
\end{theorem}

\begin{proof}
From (\ref{ecrn})  and assuming that condition
(\ref{condicion}) holds, we obtain
\begin{equation}
\label{eq:u++b}
  \Delta u^+_{\mu}=
   -\frac{q^2}{4r^4\left(  \Phi_\mu^+ \right)^3}\leq-
\frac{9J^2\sin^2\theta}{q^24r^4   \left(  \Phi_\mu^+ \right)^3  },
\end{equation}
Then, since  $m\geq |q|$ we have
\begin{equation}
  \label{eq:10}
 \left(  \Phi_\mu^+ \right)^4  \geq \left(1+ \frac{|q|}{r}\right)^2\geq\frac{q^2}{r^2},
\end{equation}
which gives us
\begin{equation}
\label{eq:u++}
\Delta u^+_{\mu}
\leq-\frac{9J^2\sin^2\theta}{4r^6  \left(  \Phi_\mu^+ \right)^{7}}=
F(x,\Phi^+_{\mu}).
\end{equation}
Now, we define the difference
\begin{equation}
  \label{eq:6}
  w= u^+_\mu- u_\mu,
\end{equation}
Using equation (\ref{ecbye})  and (\ref{eq:u++})  we obtain
\begin{equation}
  \label{eq:7}
  \Delta w \leq  F(x,\Phi^+_\mu)- F(x,\Phi_\mu).
\end{equation}
We use formula (\ref{eq:Fdif}) to conclude that
\begin{equation}
  \label{eq:8}
   \Delta w -w H(\Phi^+_\mu, \Phi_\mu) \leq 0.
\end{equation}
Note that the function $w$ is not $C^2$ at the origin because
$u^+_\mu$ is not $C^2$ there. However, the function $w$ is a weak
solution of (\ref{eq:8}) also at the origin. And hence we can apply a
weak version of the Maximum Principle for non-bounded domains. See,
for example, \cite{Maxwell04}. In this reference, the maximum
principle is proven for $H^2$ solutions. Our function $w$ satisfies
this property and since $w$ goes to zero as $r\to \infty$, we conclude
that $w\geq 0$. 

As a side comment, we note that in order to apply the maximum
principle it is only required that $w\in H^1$. Suitable versions of
the maximum principle for non-bounded domains for $H^1$ solutions can
be deduced from the weak maximum principle for bounded domains given
in \cite{Gilbarg}.
\end{proof}

Since $F$ is non-decreasing, once an upper bound is found for the sequence, the
construction of a lower bound is straightforward. Namely,
we define $u^-_\mu(x)$ as the solution of the following linear
Poisson equation
\begin{equation}
\label{eq:subrn}
  \Delta u^-_\mu  =F(x,\Phi^+_\mu)
  =- \frac{9J^2\sin^2\theta}{4r^6\left(1+\frac{m}{r}+
\frac{\mu^2}{4r^2}\right)^{7/2}},
\end{equation}
with the fall off condition
\begin{equation}
  \label{eq:28}
 \lim_{r\to \infty} u^-_\mu=0.
\end{equation}

\begin{lemma}
\label{l:rnsub}
Let $u^-_\mu$ be the solution of  (\ref{eq:subrn}) with the asymptotic
condition \eqref{eq:28}.  We have that for all $\mu > 0$
\begin{equation}
  \label{eq:14}
 u^-_\mu(x) \leq  u_\mu(x),
\end{equation}
and
\begin{equation}
  \label{eq:33}
  \frac{\mu}{2r} + u^-_\mu(x)  \geq  u^-_0(x).
\end{equation}
The  function $u^-_0$ has the following behavior
\begin{align}
  \label{eq:15}
  u^-_0(x) &= \frac{C_1}{r}+O(r^{-2}), \quad \text{ as } r\to \infty,\\
  u^-_0(x) &=\frac{C_2}{\sqrt{r}}+O(1), \quad \text{ as } r\to 0.  \label{eq:15b}
\end{align}
where $C_1,C_2>0$.
\end{lemma}
\begin{proof}
The solution  can be explicitly constructed  using the fundamental
solution (or Green function) of the Laplacian (see the Appendix). From the standard elliptic
estimates (or directly from the explicit expression) we deduce that
$u^-_\mu \in C^{2,\alpha}(\Rt)$ for $\mu>0$.

Let us prove inequality (\ref{eq:14}).  As usual we take the
difference $w=u_\mu-u^-_\mu$, then, using equation (\ref{eq:subrn}) we
have
\begin{equation}
  \label{eq:16}
  \Delta w = F(x, \Phi_\mu)- F(x, \Phi^+_\mu)=(u_\mu-u^+_\mu)H(\Phi_\mu,\Phi^+_\mu).
\end{equation}
Since $u_\mu-u^+_\mu\leq 0$ by lemma \ref{l:rn} we obtain  $\Delta
w\leq 0$ and then by the maximum principle we get $w\geq 0$.

To prove inequality (\ref{eq:33}) we use a similar argument
as in the proof of lemma \ref{l:monphi}. Note that we can in principle
deduce  (\ref{eq:33}) from the explicit expression for $u^-_\mu$,
however the formula is so complicated that this is not straightforward.

Finally, the fall-off behavior \eqref{eq:15}--\eqref{eq:15b} is
obtained from the explicit expression of $u^-_0$ given in the
Appendix (see equation
(\ref{eq:34}) and \eqref{eq:29}).

\end{proof}

Note that the sequence $u^-_\mu$ is monotonic in $\mu$, as the
Bowen-York sequence $u_\mu$. Namely, for
$\mu_1\geq\mu_2\geq 0$, we obtain
\begin{equation}
u^-_{\mu_1}(x)\leq u^-_{\mu_2}(x)\leq u^-_{0}(x),
\end{equation}
and also for  $\mu_1>\mu_2 > 0$ we  have
\begin{equation}\label{ordenfimenos}
\Phi^-_{\mu_1}(x)\geq \Phi^-_{\mu_2}(x)\geq \Phi^-_{0}(x)
\end{equation}
where $\Phi^-$ is defined as
\begin{equation}
\label{eq:30}
  \Phi^-_\mu=1+\frac{\mu}{2r}+u^-_\mu.
\end{equation}

\section{Convergence}
\label{sec:convergence-sequence}
In this section we prove that the sequence $u_\mu$ converges in the
limit $\mu \to 0$. We begin with the interior convergence. We will
make use of Lebesgue spaces $L^2$ and Sobolev spaces $H^2$ (for
definition and properties of these functional spaces see, for example,
\cite{Gilbarg}).

\begin{lemma}
\label{l:icon}
Let $U$ be an arbitrary open ball contained in $\Rt\setminus
\{0\}$. Then the sequence $u_\mu$ converges in the $H^2(U)$
norm. Moreover, the limit function
\begin{equation}
  \label{eq:12}
u_0 =\lim_{\mu \to 0} u_\mu,
\end{equation}
is a solution of equation (\ref{ecbye}) in $U$ and $u_0\in C^\infty(U)$.
\end{lemma}
\begin{proof}
Let $U'$ be an open ball contained in $\Rt\setminus
\{0\}$ such that $U\subset \subset U'$.
Let $x\in U'$ be an arbitrary but fixed point. Consider the sequence of real
numbers  $u_\mu(x)$ for $\mu \to 0$.
By lemma \ref{l:monu} the sequence  is non-decreasing  and by
lemma \ref{l:rn} it is bounded from above by  $u_\mu(x) \leq
u^+_0(x)$. Note that it is important that the closure of
$U'$ does not contain the origin $\{0\}$, since $u^+_0$ is not bounded there.
It follows that the sequence converges pointwise to a limit
$u_0(x)$. To prove convergence in Lebesgue norm we use the Dominated
Convergence Theorem (see e.g. \cite{Evans98}). In particular, this
implies that the sequence converges in $L^2(U')$, i.e. the sequence
$u_\mu$ is Cauchy in $L^2(U')$
\begin{equation}
\label{limw}
\lim_{\mu_1,\mu_2\rightarrow 0}||w||_{L^2(U')}=0,
\end{equation}
 where $w=u_{\mu_2}-u_{\mu_1}$.

 To prove that the sequence $u_{\mu}$ is a Cauchy sequence  in
 $H^2(U)$ we use the standard elliptic estimate for the Laplacian (see
 e.g. \cite{Gilbarg})
\begin{equation}
\label{estimacion}
||w||_{H^2(U)}\leq C\left(||\Delta w||_{L^2(U')}+||w||
_{L^2(U')}\right)
\end{equation}
where the constant $C$ depends only on $U'$ and
$U$.

The difference $w$ satisfies equation (\ref{ww}), then we obtain
\begin{align}
\label{deltaa}
||\Delta w||_{L^2(U')} &=
\left \Vert Hw+H\frac{\mu_2-\mu_1}{r}\right\Vert_{L^2(U')},\\
& \leq ||Hw||_{L^2(U')}+(\mu_1-\mu_2)\left\Vert\frac{H}{r}\right\Vert_{L^2(U')}.
\end{align}
The functions $H$ and $H/r$ are bounded in $U'$ (see equation
(\ref{eq:13})) by a constant independent of $\mu$. Then, from the
inequality (\ref{deltaa}) we obtain
\begin{equation}
\label{delta1a}
||\Delta w||_{L^2(U')}\leq
C\left(||w||_{L^2(U')}+(\mu_1-\mu_2)\right).
\end{equation}
where $C$ does not depend on $\mu$. Using the estimate (\ref{estimacion})
we finally get
\begin{equation}
  ||w||_{H^2(U)}\leq C\left(||w||_{L^2(U')}+(\mu_1-\mu_2)\right).
\end{equation}
From this inequality and the convergence in $L^2$ given by
(\ref{limw}) we conclude that
\begin{equation}
\lim_{\mu_1,\mu_2\rightarrow 0}||w||_{H^2(U)}=0.
\end{equation}
and hence $u_{0}=\Phi_0-1\in H^2(U)$. By the same argument we also have that
$u_0$ is a strong solution (see
\cite{Gilbarg} for the definition of strong solutions for elliptic
equations) of equation (\ref{ecbye}) in $U$.

Using the standard elliptic estimates once again and iterating
using equation  (\ref{ecbye}) we get that $u_0\in
C^\infty(U)$. This iteration can be done as follows. By the Sobolev
imbedding theorem we have that $u_0\in C^\alpha(U)$. Then, it follows
that $F(x,\Phi_{0})\in C^\alpha(U) $.  But then, by H\"older estimates
for the Laplace operator (see \cite{Gilbarg}) it follows that $u_0\in
C^{2,\alpha}(U)$. We can iterate this argument to obtain that $u_0$ is
smooth in $U$.
\end{proof}

In the previous theorem we have not analyzed the fall off of the
solution $u_0$ at infinity and its behavior at the origin. In order
to do so, more precise estimates are required. In particular, we
need to make use of weighted Sobolev norms. We will use the weighted
Sobolev spaces defined in \cite{Bartnik86} and denoted here by
$H'^{k,\delta}$. The definitions of the corresponding norms are the
following (we restrict ourselves to the case $p=2$ and dimension $3$)
\begin{equation}
\label{eq:35}
||f||'_{L'^{2,\delta}}=\left(\int_{\Rt\setminus\{0\}} |f|^2r^{-2\delta-3}dx\right)^{1/2},
\end{equation}
and
\begin{equation}
\label{eq:38}
||f||'_{H'^{k,\delta}}:=\sum_{0}^{k}||D^jf||'_{L'^{2,\delta-j}}.
\end{equation}
These functional spaces are relevant for our purpose because we have that  
\begin{equation}
\label{sobrn}
  u^+_{\mu}(x)\in H'^{2,\delta} \quad \text{ for }    -1<\delta< -1/2,
\end{equation}
for all $\mu\geq 0$.  We can understand the given range of $\delta$ by
noticing that the extreme Reissner N\"ordstrom solution goes as
$r^{-1/2}$ as $r\rightarrow0$, and as $r^{-1}$ as
$r\rightarrow\infty$. It can also be seen that, if we consider only
solutions with $\mu>0$, then the allowed interval for $\delta$ expands
to $(-1,0)$ reflecting the fact that in this case, the functions are
bounded at the origin.

\begin{lemma}
\label{l:ws}
The sequence $u_\mu$ is Cauchy in the norm $H'^{2,\delta}$ for
$-1<\delta<-1/2$.
\end{lemma}

\begin{proof}
The proof is similar as in the previous lemma, the main difference is
that we have to take into account the singular behavior of the
functions at the origin.

We first note that the same argument presented above allows us to prove
convergence in the weighted Lebesgue spaces $L'^{2,\delta}$. In
effect, consider the sequence $u_{\mu}r^{-\delta-3/2}$ for
$-1<\delta<-1/2$ . This sequence is pointwise bounded by
$u^+_{0}r^{-\delta-3/2}$ and monotonically increasing
as the parameter $\mu$ goes to zero, which means that it is
a.e. pointwise converging to a function $u_0r^{-\delta-3/2}$. Then, we
can use the Dominated Convergence Theorem (since
$u^+_{0}r^{-\delta-3/2}$ is summable in $\Rt$ for the given values of
the weight $\delta$) to find that the new sequence converges in
$L^2(\Rt)$. But this implies that the original sequence $u_{\mu}$
converges in $L'^{2,\delta}$, with $\delta\in(-1,-1/2)$. That is
\begin{equation}
\label{limwr3}
\lim_{\mu_1,\mu_2\rightarrow 0}||w||_{L'^{2,\delta}}=0,
\end{equation}
where  $w$ is the difference introduced above in equation (\ref{w}).

In order to prove that the sequence $u_{\mu}$ is a Cauchy
sequence also in the weighted Sobolev space $H'^{2,\delta}$ with
$\delta\in(-1,-1/2)$, we will apply the following estimate (see,
e.g. \cite{Bartnik86})
\begin{equation}\label{estimacion1}
||w||_{H'^{2,\delta}}\leq C||\Delta w||_{L'^{2,\delta-2}},
\end{equation}
where the constant $C$ depends only on $\delta$.

As before,  we obtain
\begin{align}
\label{delta}
  ||\Delta w||_{L'^{2,\delta-2}} & =\left\Vert
    Hw+H\frac{\mu_2-\mu_1}{r}\right\Vert_{L'^{2,\delta-2}} \\
 & \leq
 ||Hw||_{L'^{2,\delta-2}}+(\mu_1-\mu_2)\left\Vert\frac{H}{r}\right\Vert_{L'^{2,\delta-2}}.
\end{align}
From the definition of the norm $L'^{2,\delta}$ given in (\ref{eq:35})
we obtain
\begin{equation}
  \label{eq:36}
   ||Hw||_{L'^{2,\delta-2}}\leq \sup_{\Rt}|Hr^2| \,  ||w||_{L'^{2,\delta}},
\end{equation}
and hence, using (\ref{delta}) we have
\begin{equation}
\label{delta1}
||\Delta w||_{L'^{2,\delta-2}}\leq
C\left(\sup_{\Rt}|Hr^2| \,||w||_{L'^{2,\delta}}+
(\mu_1-\mu_2)\left\Vert\frac{H}{r}\right\Vert_{L'^{2,\delta-2}}\right).
\end{equation}
The crucial step in the proof is to bound, in (\ref{delta1}), the
corresponding norms of $H$ and $H/r$.
At this point is where the weighted Sobolev spaces play a role,
because these norms are not bounded in the standard Sobolev norms.

To bound $Hr^2$ we use
\begin{equation}
\label{eq:31}
H\leq \frac{63J^2\sin^2\theta}{4r^6}7\left(1+u^-_0\right)^{-8}.
\end{equation}
By theorem \ref{l:rnsub} we know that $u^-_0$ goes to zero at infinity,
hence $H$ decays as $r^{-6}$. At the origin, by lemma \ref{l:rnsub},
we know that $u^-_0=\mathcal{O}(r^{-1/2})$, therefore, $H$ grows as
$r^{-2}$. Hence the $r^2H$  is finite for
every value of the parameter $\mu$.

For the other term  we have
\begin{equation}
\left\Vert\frac{H}{r}\right\Vert_{L'^{2,\delta-2}}=
\left(\int_{\Rt\setminus\{0\}}\left|\frac{H}{r}\right|^2r^{-2\delta+1}dx\right)^{1/2}
\end{equation}
and, using again the lower bound as in \eqref{eq:31} we find that
this norm is also finite for $\delta\in(-1,-1/2)$.  Then, we can write
\begin{equation}
||w||_{H'^{2,\delta}}\leq C \left( ||w||_{L'^{2,\delta}}+(\mu_1-\mu_2)\right),
\end{equation}
where the constant $C$ does not depend on $\mu$.  This and equation
(\ref{limwr3}) give us, in the limit $\mu_1,\mu_2\rightarrow 0$
\begin{equation}
\lim_{\mu_1,\mu_2\rightarrow 0}||w||_{H'^{2,\delta}}=0.
\end{equation}
Then,  the sequence $u_{\mu}$ is
Cauchy in the $H'^{2,\delta}$ norm, with
$\delta\in(-1,-1/2)$.
\end{proof}

Note that this theorem also implies that $u_0$ is a strong solution in
the Sobolev spaces $H'^{2,\delta}$ of equation (\ref{ecbye}) also at
the origin.
\section{Final comments}
\label{sec:final-comments}
In this article, we have studied the extreme limit of the Bowen-York
family of initial data. We have found that the extreme solution exists and
has similar properties to the extreme Kerr black hole data.  It is
straightforward to generalize the results presented here for more
general second fundamental forms keeping the conformal flatness of the
data. A more relevant and difficult generalization would involve more
general background metric. In particular, it would be interesting to
generalize the extreme limit for binary Kerr black hole data. A
possible strategy to attack this problem is to prove, using similar
techniques as the ones presented here, that the sequence
of two non-extreme Kerr black holes constructed in \cite{Avila:2008te}
actually converges in the extreme limit.

As it was mentioned in the introduction, there exists a variational
characterization of the extreme limit. The extreme initial data, and
hence data with cylindrical ends, appears naturally as minimum of the
mass in appropriate class of data. The example presented here
incorporates a new class of data in which this variational
characterization holds. As we said in section
\ref{sec:main-result}, we expect that this minimum of the mass
(i.e. the extreme solution) has no horizon. Moreover, we expect that a
small perturbation of an extreme solution (in particular, the extreme
Bowen-York data) will always have an horizon.  It would be interesting
to prove or disprove this conjecture.

\appendix
\section{Explicit expression of the subsolution $u^-_\mu(x)$}

In this section  we construct the explicit solution to equation
\begin{equation}
\Delta u_-(x,\mu)=F.
\end{equation}
where $F$ is given by
\begin{equation}
  \label{eq:20}
  F=\sin^2\theta R(\mu,r),
\end{equation}
and
\begin{equation}
  \label{eq:17}
  R(\mu,r)=-\frac{18J^2}{8r^6\left(1+\frac{m}{r}+
\frac{\mu^2}{4r^2}\right)^{7/2}}.
\end{equation}
The solution is constructed integrating the Green function of the Laplacian, that is
\begin{equation}
  \label{eq:18}
 u_-(x,\mu)=-\int_{\Rt} \frac{F(x')}{|x-x'|}\, dx'.
\end{equation}
We use the expansion of the Green function in terms of spherical
harmonics (see, for example, \cite{jackson99}). The angular dependence
of the source $F$ is given by $\sin^2\theta$, which has an expansion
in terms of the following two spherical harmonics
\begin{equation}
  \label{eq:sphar}
  Y_{00}=\frac{1}{\sqrt{4\pi}}, \quad
  Y_{20}=\sqrt{\frac{5}{16\pi}}(3\cos^2\theta-1),
\end{equation}
namely
\begin{equation}
  \label{eq:19}
  \sin^2\theta= \frac{2}{3}\sqrt{4\pi} \left(Y_{00} - \frac{Y_{20}}{\sqrt{5}}\right).
\end{equation}
Hence, it follows that the angular dependence of the solution can also
be expanded in term of these  two spherical harmonics. That is,
$u^-_\mu$ has the form (\ref{eq:subsolm})
where the radial functions $I_1(r)$ and $I_2(r)$ are given by the
following integrals
\begin{align}
I_1 &=\int_{0}^r R(r',\mu)\frac{1}{r}r'^2dr'+\int_{r}^\infty R(r',\mu)\frac{1}{r'}r'^2dr',\\
I_2 &=\int_{0}^r R(r',\mu)\frac{r'^2}{r^3}r'^2dr'+\int_{r}^\infty R(r',\mu)\frac{r^2}{r'^3}r'^2dr'.
\end{align}
Computing these integrals, we find
\begin{multline}
\label{eq:45}
I_1=\frac{2\sqrt{\pi} J^2}{5rq^6}\left(-8(4\mu^2+3q^2)(2r+\mu)+\right.\\
+\frac{(4\mu^2+3q^2)(16r^4+\mu^4)+4mr(5q^2+8\mu^2)(4r^2+\mu^2)+}{(r^2+mr+\frac{\mu^2}{4})^{3/2}}\\
\left.+\frac{+6r^2(5q^4+16\mu^4+20q^2\mu^2)}{(r^2+mr+\frac{\mu^2}{4})^{3/2}}\right),
\end{multline}
and
\begin{multline}
\label{eq:46}
I_2=\frac{2\sqrt{\pi} J^2}{5r^3q^6}\left(-8(2r+\mu)(16r^4-8r^3\mu+4r^2\mu^2-2r\mu^3+\mu^4)+\right.\\
+\frac{256r^8+\mu^8+6rm(64r^6+\mu^6)
  +96r^6(q^2+2\mu^2)}{(r^2+mr+\frac{\mu^2}{4})^{3/2}}+\\
\left.+\frac{4r^3m(2\mu^2-q^2)(4r^2+\mu^2)+6r^2(2\mu^6+\mu^4q^2+r^2q^4)}{(r^2+mr+\frac{\mu^2}{4})^{3/2}}\right). 
\end{multline}
From these expressions we see that
\begin{equation}
\label{eq:43}
u^-_\mu = \frac{64J^2}{5r(2m+\mu)^3} + O(r^{-2}) \quad r\rightarrow\infty,
\end{equation}
and
\begin{equation}
  \label{eq:44}
  u^-_\mu(r=0) = \frac{4J^2(2m-\mu)^3}{5\mu q^6}.
\end{equation}
When $\mu=0$ the radial  functions (\ref{eq:45})--(\ref{eq:46}) reduce to
\begin{align}
\label{eq:I1mu1}
  I_1|_{\mu=0} &=\frac{4J^2\sqrt{\pi}}{5q^4\sqrt{r}}
\left(\frac{24r^2+40qr+15q^2}{(r+q)^{3/2}}-24\sqrt{r}\right),\\
\label{eq:I2mu2}
  I_2|_{\mu=0} &=\frac{4J^2\sqrt{\pi}}{5q^6\sqrt{r}}
\left(\frac{128r^4+192r^3q+48r^2q^2-8q^3r+3q^4}{(r+q)^{3/2}}-128r^{5/2}\right).
\end{align}
In this case the asymptotic behaviors are given by
\begin{equation}
\label{eq:34}
u^-_0 = \frac{8J^2}{5rq^3} + O(r^{-2})  \quad r\rightarrow\infty,
\end{equation}
and
\begin{equation}
\label{eq:29}
u^-_0 = \frac{9J^2(17-\cos^2\theta)}{25q^{7/2}\sqrt{r}}  + O(1)\quad
r\rightarrow 0.
\end{equation}
Finally, we mention that it is possible to construct a positive  lower bound which is
spherically symmetric and has the correct behavior at the origin and
at infinity. Namely, from (\ref{eq:subsolm}) we deduce
\begin{equation}
\label{cotaumenos}
u^-_\mu\geq Y_{00}\left(I_1-\frac{1}{5}I_2\right)\geq0.
\end{equation}


\end{document}